\documentclass[thm-restate]{lipics-v2021}
\usepackage[utf8]{inputenc}
\usepackage{url}
\usepackage[numbers,sort]{natbib}

\nolinenumbers

\title{Computing a Subtrajectory Cluster from $c$-packed Trajectories}
\author{Joachim Gudmundsson}{The University of Sydney}{joachim.gudmundsson@sydney.edu.au}{https://orcid.org/0000-0002-6778-7990
}{Funded by the Australian Government through the Australian Research Council DP180102870.}

\author{Zijin Huang}{The University of Sydney}{zijin.huang@uni.sydney.edu.au}{https://orcid.org/0000-0003-3417-5303}{}

\author{André van Renssen}{The University of Sydney}{andre.vanrenssen@sydney.edu.au}{https://orcid.org/0000-0002-9294-9947}{}
\date{October 2022}

\author{Sampson Wong}{The University of Sydney}{swon7907@sydney.edu.au}{https://orcid.org/0000-0003-3803-3804}{}

\acknowledgements{The authors thank Kevin Buchin for the insightful discussion on the important Lemma~\ref{lemma:strip_limit}. The authors thank the anonymous reviewers for their helpful feedback.}

\Copyright{Joachim Gudmundsson, Zijin Huang, André van Renssen and Sampson Wong}
\authorrunning{J.\, Gudmundsson, Z.\, Huang, A.\, van Renssen and S.\, Wong}
\ccsdesc[10010063]{Computational Geometry}
\keywords{Subtrajectory cluster, $c$-packed trajectories, Computational geometry}
  
\newtheorem{problem}{Problem}
\newtheorem{subproblem}[problem]{Subproblem}
\newtheorem{fact}[theorem]{Fact}

\DeclareUnicodeCharacter{2212}{-}

\theoremstyle{remark}

\newcommand{\abse}[1]{\lVert #1 \rVert}
\newcommand{\abs}[1]{|#1|}

\newcommand{\ed}{\varepsilon d}
\newcommand{\e}{\varepsilon}
\newcommand{\ehat}{\hat{\varepsilon}}

\makeatletter
\def\NAT@spacechar{~}% NEW
\makeatother

\usepackage{thmtools}
\usepackage{thm-restate}
\usepackage{proof-at-the-end}

\newcommand{\Frechet}{Fréchet }

\begin{document}
\maketitle

% variable for compicated running times
\newcommand{\finalComplexity}{$O((c^2 n/\varepsilon^2)\log(c/\varepsilon)\log(n/\varepsilon))$ }
\newcommand{\constructFdComplexity}{(cn/\varepsilon) \log{(cn/\varepsilon)} }

\begin{abstract}
We present a near-linear time approximation algorithm for the subtrajectory cluster problem of $c$-packed trajectories. The problem involves finding $m$ subtrajectories within a given trajectory $T$ such that their Fréchet distances are at most $(1 + \varepsilon)d$, and at least one subtrajectory must be of length~$l$ or longer. A trajectory $T$ is $c$-packed if the intersection of $T$ and any ball $B$ with radius $r$ is at most $c \cdot r$ in length.

Previous results by Gudmundsson and Wong \cite{GudmundssonWong2022Cubicupperlower} established an $\Omega(n^3)$ lower bound unless the Strong Exponential Time Hypothesis fails, and they presented an $O(n^3 \log^2 n)$ time algorithm. We circumvent this conditional lower bound by studying subtrajectory cluster on $c$-packed trajectories, resulting in an algorithm with an \finalComplexity time complexity.
\end{abstract}

\section{Introduction}
With the proliferation of location-aware devices comes an abundance of trajectory data. One way to process and make sense of many trajectories is to group long and similar subtrajectories. The analysis of long and similar parts of trajectories can provide insights into behavior and mobility patterns, such as common routes taken and places visited frequently. 

Buchin et al. \cite{BuchinEtAl2011Detectingcommutingpatterns} initialised the study of subtrajectory cluster problems to detect and extract common movement patterns. The Subtrajectory Cluster (SC) decision problem is defined as follows. Given one or more trajectories, determine if there exists a cluster of $m - 1$ non-overlapping subtrajectories and one reference trajectory. The reference trajectory~$T_r$ must be at least of length $l$, and the Fréchet distances between $T_r$ and the other $m - 1$ subtrajectories must be at most $d$. In the case of animals, long and common movement patterns can indicate movement between grazing spots of sheep or the migration flyway of seabirds. In the case of humans, common movement on a Monday morning can show commuting patterns to find the most heavily congested areas. 

Subtrajectory clustering has attracted research from multiple communities. Gudmundsson and Wolle~\cite{gudmundssonFootballAnalysisUsing2012} used subtrajectory cluster to analyse the common movement patterns of football players. Buchin et al.~\cite{buchinClusteringTrajectoriesMap2017} applied subtrajectory cluster to map reconstruction by clustering common movement patterns of vehicles into road segments. Researchers in the Geographical Information and Data Mining communities also considered the variants and practical performance of subtrajectory cluster algorithms \cite{changMultigranularityVisualizationTrajectory2009, tampakisScalableDistributedSubtrajectory2019, gudmundssonMotifsGoalsMining2012, gudmundssonGPUApproachSubtrajectory2015, leeTrajectoryClusteringPartitionandgroup2007, agarwalSubtrajectoryClusteringModels2018, gudmundssonPracticalNearestSubTrajectory2021}. In addition, the potential of SC is examined in a wide range of applications, including sports player analysis~\cite{wangSimilarSportsPlay2021} and human movement analysis~\cite{cavallaroMeasuringImpactCOVID192021, hosseinpoormilaghardanActivitybasedFrameworkDetecting2021}. 

Several theoretical studies of the subtrajectory clustering problem focus on improving the quality of clustering. Agarwal et al.~\cite{agarwalSubtrajectoryClusteringModels2018} defined a single objective function, the weighted sum of three quality measures of a clustering. These quality measures include the number of clusters chosen, the quality of the cluster, and the size of the trajectories excluded from the clustering. Brüning et al.~\cite{bruningFasterApproximateCovering2022} studied so-called $\triangle$-coverage, aiming to find a set $C$ of curves to cover a polygonal curve such that a curve in $C$ is fixed in size, and $\abs{C}$ is minimised. 

However, despite considerable attention from multiple communities, there is no subcubic time algorithm that solves the subtrajectory cluster problem, limiting its usefulness on large data sets. Buchin et al.~\cite{BuchinEtAl2011Detectingcommutingpatterns} solved the subtrajectory cluster problem in $O(n^5)$ time when the similarity measurement of two trajectories is the Fréchet distance, and Gudmundsson and Wong~\cite{GudmundssonWong2022Cubicupperlower} further improved the runtime with an $O(n^3 \log^2n)$ time algorithm. In addition, Gudmundsson and Wong~\cite{GudmundssonWong2022Cubicupperlower} showed that there is no $O(n^{3-\delta})$ algorithm for subtrajectory cluster for any $\delta > 0$ unless the Strong Exponential Time Hypothesis (SETH) fails. 

SC is unlikely to have a strongly subquadratic algorithm even if we allow a small approximation factor on the Fr\'echet distances between subtrajectories. Because given two trajectories $T_1$ and $T_2$, we can structure the SC problem to find two subtrajectories such that their Fréchet distance is at most $(1 + \varepsilon)d$, and the reference trajectory must be as long as the maximum of $T_1$ and $T_2$. Solving this instance of SC is equivalent to approximating the Fréchet distance of $T_1$ and $T_2$, and Bringmann~\cite{bringmannWhyWalkingDog2014} showed that there is no $1.001$-approximation with runtime $O(n^{2 − \delta})$ for the continuous Fréchet distance for any $\delta > 0$, unless SETH fails.

Since an exact subcubic and an approximate subquadratic algorithm are unlikely to exist, we study subtrajectory cluster on a realistic family of trajectories, called $c$-packed trajectories. A trajectory $T$ is $c$-packed if, for any ball $B$ of radius $r$, the length of $T$ lying inside $B$ is at most $c$ times $r$. The packedness value of a trajectory $T$ is the maximum $c$ for which $T$ is $c$-packed. Bringmann~\cite{bringmannWhyWalkingDog2014} proved that computing the Fréchet distance has no strong subquadratic algorithm unless SETH fails, and the notion of $c$-packedness was introduced by Driemel et al.~\cite{DriemelEtAl2012ApproximatingFrechetDistance} to circumvent such conditional lowerbound. Since then, the notion of $c$-packedness has gained considerable attention from the theory community~\cite{gudmundssonFrechetQueriesGeometric2013, bringmannWhyWalkingDog2014, aghamolaeiWindowingQueriesUsing2021, chenApproximateMapMatching2011, bringmannImprovedApproximationFrechet2017}, and several real-world data sets have been shown to have low packedness values~\cite{gudmundssonApproximatingPackednessPolygonal2023, DriemelEtAl2012ApproximatingFrechetDistance}. In one particular instance,  \mbox{Gudmundsson et al.~\cite{gudmundssonApproximatingPackednessPolygonal2023}} approximated the packedness values of several real-world trajectory data sets. In their experiments, several trajectory data sets  have low packedness values, such as the movement patterns of people in Beijing, school buses, European football players, and trawling bats. 

\sloppy In this paper, given a $c$-packed trajectory $T$ of complexity $n$ and a desired multiplicative approximation error $\varepsilon$ on the Fréchet distance between subtrajectories, we present an \finalComplexity time algorithm that solves the SC problem. It is worth noting that previous papers considering $c$-packed curves typically replace a factor $n$ with a polynomial of constant degree in $c$ \cite{conradiComputingKShortcutFrechet2022, gudmundssonMapMatchingQueries2023}. We are able to replace a factor of $n^2$ with $c^2/\varepsilon^2$, bringing the algorithm's running time from cubic to near-linear, assuming $c \in O(1)$.

Along the way, we develop a tool for simplifying the free space diagram that may be of independent interest. To efficiently approximate the \Frechet distance, Driemel et al.~\cite{DriemelEtAl2012ApproximatingFrechetDistance} showed that the free space complexity, i.e., the number of non-empty cells, is $O(cn/\e)$ for two simplified $c$-packed trajectories (see Section~\ref{sec:overview} for an overview of the free space, or \cite{altComputingFrechetDistance1995} for a formal definition). However, simplifying a trajectory by taking shortcuts between vertices can yield a much shorter trajectory, and the SC problem is sensitive to the length of the trajectories since the reference trajectory has to have a length at least $l$. To tackle this problem, we developed a tool to construct the free space diagram in~$O(\constructFdComplexity)$ time, preserving the length of two trajectories while benefiting from the $O(cn/\e)$ free space complexity. Our tool can be of value for problems in which the length of a trajectory is important, such as subtrajectory cluster \cite{BuchinEtAl2011Detectingcommutingpatterns}, partial curve matching \cite{buchinExactAlgorithmsPartial2009}, and \Frechet distance with speed limit \cite{maheshwariFrechetDistanceSpeed2011}.

In Section~\ref{sec:overview}, we will formally define the subtrajectory cluster problem and outline the greedy plane sweep algorithms by Buchin et al.~\cite{BuchinEtAl2011Detectingcommutingpatterns} and Gudmundsson and Wong~\cite{GudmundssonWong2022Cubicupperlower}, which our approach builds on. In Section~\ref{sec:technical_overview}, we provide a technical overview of our main results. 
In Section~\ref{sec:free_space_diagram}, we will discuss how to simplify the free space diagram to achieve a lower complexity while preserving trajectory lengths.

In Section~\ref{sec:vertex_to_vertex}, we will consider the restricted case when the reference trajectory must be vertex-to-vertex. In Section~\ref{sec:arbitrary_ref_trajectory}, we will remove this restriction by considering arbitrary reference trajectory.

\section{Preliminaries}
\label{sec:overview}
In this section, we will outline the previous algorithms for the subtrajectory cluster problem. The subtrajectory cluster problem was first introduced by Buchin et al.~\cite{BuchinEtAl2011Detectingcommutingpatterns}, and later improved by Gudmundsson and Wong~\cite{GudmundssonWong2022Cubicupperlower}. But instead of looking for a subtrajectory where the \Frechet distances between the reference trajectory and the subtrajectories are exact, we aim to find a solution that approximates the \Frechet distance between subtrajectories in the cluster. 

\begin{problem}[{{\cite{GudmundssonWong2022Cubicupperlower}}}]
Given trajectory $T$ of complexity $n$, a positive integer $m$, positive real numbers $d$, $l$ and $\e$, decide if there exists a subtrajectory cluster of $T$ such that:
\begin{itemize}
    \item the cluster consists of one reference subtrajectory and $m - 1$ other subtrajectories of $T$, 
    \item the reference subtrajectory has Euclidean length at least $l$, 
    \item the Fréchet distance between the reference subtrajectory and any other subtrajectory is at most $(1 + \e)d$, 
    \item any pair of subtrajectories in the cluster overlap in at most one point.
\end{itemize}
\end{problem}

Buchin et al.~\cite{BuchinEtAl2011Detectingcommutingpatterns} solved the exact SC problem by using a plane sweep algorithm on the free space diagram $\mathcal{F}_d(T, T)$. Let $s$ and $t$ be two points on $T$, and we denote $T_{st}$ as the subtrajectory of $T$ starting from $s$ and ending on $t$. Let $l_s$ and $l_t$ be the vertical sweep lines $x = s$ and $x = t$ on $\mathcal{F}_d(T, T)$, respectively (see Figure~\ref{fig:free_space_and_paths}). An $xy$-monotone path in $\mathcal{F}_d(T, T)$, or monotone path for short, is a continuous path that is non-decreasing in both $x$- and $y$-coordinates. To solve $\texttt{SC(m, d, l)}$, the lines $l_s$ and $l_t$ sweep from left to right while making sure that $l_s$ is to the left of $l_t$, and the reference trajectory $T_{st}$ is at least $l$ long, i.e., $t - s \geq l$. In each interval $[l_s, l_t]$, they compute the maximum number of monotone paths in $\mathcal{F}_d(T, T)$ starting at $l_s$ and ending at $l_t$. 

\begin{figure}[bth]
    \centering
    \includegraphics[scale=0.9]{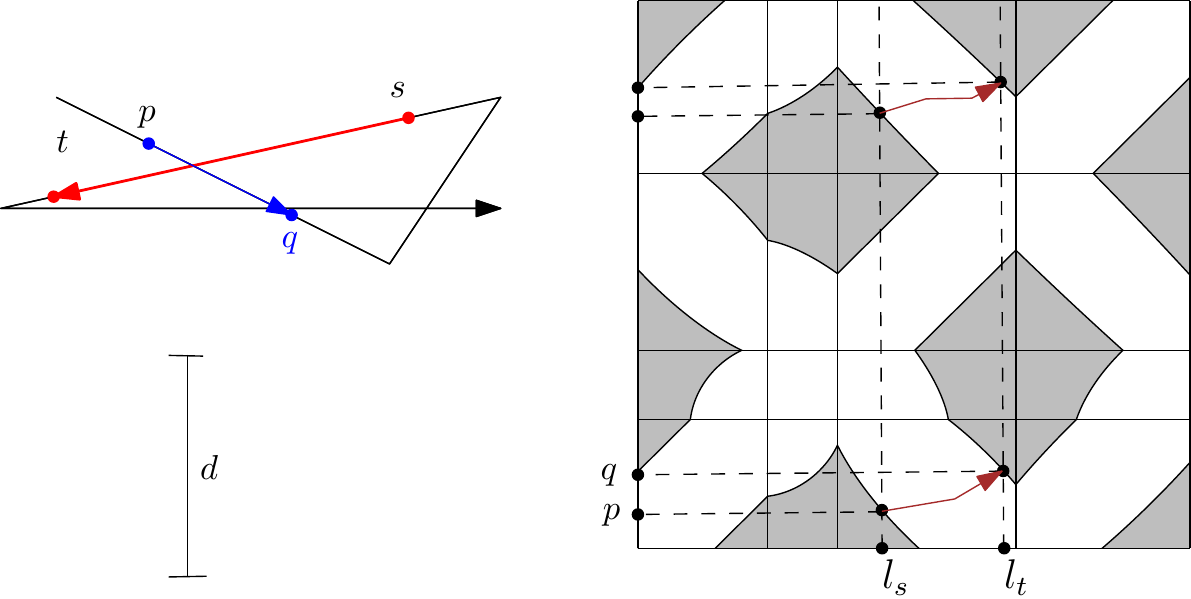}
    \caption{The free space diagram $\mathcal{F}_d(T, T)$, and interval $[l_s, l_t]$ defined by two points $s$ and $t$ on~$T$. If there exists a monotone path (marked in brown) from $(s, p)$ to $(t, q)$ through the free space, then the Fréchet distance between subtrajectories $T_{st}$ and $T_{pq}$ is at most $d$. }
    \label{fig:free_space_and_paths}
\end{figure}

Let $p$ and $q$ be two points on $T$, and let $(s, p)$ and $(t, q)$ be two coordinates on $\mathcal{F}_d(T, T)$. As we only consider monotone paths starting from $l_s$ and ending on $l_t$, we call the monotone path from $(s, p)$ to $(t, q)$ the $pq$ monotone path. First, a monotone path $pq$ must traverse only the free space. Second, two monotone paths $pq$ and $ab$ must not overlap along the $y$-interval in more than a single point. Third, the $y$-coordinates of any $pq$ monotone path cannot overlap the $[s, t]$ interval in more than a single point. We obtain the following subproblem.

\begin{subproblem}[\cite{GudmundssonWong2022Cubicupperlower}]
\label{problem:2}
Given a trajectory $T$ of complexity $n$, a positive integer $m$, a positive real value $d$, and a reference subtrajectory of $T$ starting at $s$ and ending at $t$, let $l_s$ and $l_t$ be two vertical lines in $\mathcal{F}_d(T, T)$ representing the points $s$ and $t$. Decide if there exist:
\begin{itemize}
    \item $m - 1$ distinct paths starting at $l_s$ and ending at $l_t$, such that 
    \item the $y$-coordinate of any two monotone paths overlap in at most one point, and
    \item the $y$-coordinate of any monotone path overlaps the $y$-interval from $s$ to $t$ in at most one point.
\end{itemize}
\end{subproblem}

To look for a set $\{p_1q_1, p_2q_2, ..., p_{m - 1}q_{m - 1}\}$ of monotone paths, both algorithms use a greedy approach. First, set $p_1$ to be the lowest feasible point on $l_s$, and compute $p_1q_1$ by searching for a lowest monotone path through the free space. Inductively, with $p_{i - 1}q_{i - 1}$ computed, set $p_{i}$ to the lowest feasible point on $l_s$ that is on or above $q_{i - 1}$, and do the same. If a search from $p_i$ leads to a dead end, we simply set $p_i$ to the next lowest feasible point on $l_s$, and search again.

The sweeplines stop at all $O(n^3)$ critical points, and for each critical point there is a $[l_s, l_t]$ interval to consider. Buchin et al.~\cite{BuchinEtAl2011Detectingcommutingpatterns} solved each instance in $O(nm) \subseteq O(n^2)$ time. Gudmundsson and Wong~\cite{GudmundssonWong2022Cubicupperlower} improved the efficiency by connecting the critical points efficiently in a tree-like data structure which allows them to reuse computed monotone paths from previous interval instances. They showed that, in their construction, there are at most $O(n^3 \log n)$ edges, and each edge takes at most $O(\log n)$ time to add, remove, or access. This brings down the complexity of the algorithm from $O(n^5)$ to $O(n^3 \log^2 n)$ time.

\section{Technical Overview}
\label{sec:technical_overview}

Our technical overview is divided into three parts. In Sections~\ref{sec:techview:free_space_diagram_construction}, ~\ref{sec:techview:vertex_to_vertex}, and~\ref{sec:techview:arbitrary_reference}, we summarise the main result of Sections~\ref{sec:free_space_diagram},~\ref{sec:vertex_to_vertex}, and~\ref{sec:arbitrary_ref_trajectory} respectively. 

\subsection{Computing the Free Space Diagram}
\label{sec:techview:free_space_diagram_construction}

Our algorithm constructs a simplified free space diagram that preserves trajectory lengths. The size (in terms of Euclidean length) of the simplified free space diagram is the same as the size of the unsimplified free space diagram. The only difference between the two diagrams is that approximate distances are used in the simplified diagram. In particular, we define a function that uniformly maps a trajectory to its simplification, and we calculate the distance between the mapped simplification points instead of points on the original trajectory. We prove that the complexity of the simplified free space diagram will be at most $O(cn / \varepsilon)$, and that the trajectory lengths in the diagram are preserved. Next, we build the simplified free space diagram. We use an algorithm by Conradi and Driemel~\cite{conradiComputingKShortcutFrechet2022} to query pairs of nearby segments. Finally, we construct a data structure on the free space diagram so that we can access the closest non-empty cells below, above, to the left, and to the right in constant time. Putting this all together, we obtain Theorem~\ref{thm:free_space_summarised}. For a full proof, see Section~\ref{sec:free_space_diagram}.

\begin{restatable}{theorem}{freespacesummarised}
\label{thm:free_space_summarised}
Given a pair of trajectories, one can construct a simplified free space diagram in~$O(\constructFdComplexity)$ time, so that the simplified free space has complexity~$O(cn / \varepsilon)$, it approximates the Fr\'echet distance to within a factor of~$(1+\varepsilon)$, and it preserves the trajectory lengths of the original trajectory.
\end{restatable}

\subsection{Reference trajectory is vertex-to-vertex}
\label{sec:techview:vertex_to_vertex}

Next, we focus on the special case where the reference trajectory is vertex-to-vertex. Three data structures are used in the vertex-to-vertex subtrajectory cluster algorithm of Gudmundsson and Wong~\cite{GudmundssonWong2022Cubicupperlower} --- a directed graph, a range tree, and a link-cut tree. For an overview of these data structures, see Appendix~\ref{appendix:graph_construction}. Originally, the number of leaves per range tree is~$O(n)$, and the directed graph has complexity~$O(n^2)$. We use the $c$-packedness property to prove that, in our simplified free space diagram, the number of leaves per range tree is~$O(c/\varepsilon)$, and the directed graph has complexity~$O((cn/\varepsilon) \log (c/\varepsilon))$. The link-cut tree data structure can be used without modification. Putting this all together, we obtain Theorem~\ref{thm:vertex_to_vertex}. Recall that~$m$ is the desired number of subtrajectories in the cluster. For a full proof, see Section~\ref{sec:vertex_to_vertex}. 

\begin{restatable}{theorem}{vertextovertex}
\label{thm:vertex_to_vertex}
There is an $O(nm \log(c/\varepsilon) \log(n/\varepsilon))$ time algorithm that solves $\mathtt{SC(T, m, l, (1 + \varepsilon)d)}$ in the case that the reference trajectory is vertex-to-vertex.
\end{restatable}

\subsection{Reference trajectory is arbitrary}
\label{sec:techview:arbitrary_reference}

Finally, we tackle the general case where the reference trajectory is arbitrary. The main obstacle in the general case is that there are~$\Theta(n^3)$ internal critical points that correspond to potential starting and ending positions of the reference trajectory. In fact, Gudmundsson and Wong~\cite{GudmundssonWong2022Cubicupperlower} show that, for general (not $c$-packed) curves, these internal critical points are essentially unavoidable. They use the~$\Theta(n^3)$ internal critical points to show that under the Strong Exponential Time Hypothesis (SETH), there is no~$O(n^{3 - \delta})$ time algorithm for subtrajectory cluster for any~$\delta > 0$.

Our main lemma in this section is to bound the number of internal critical points for subtrajectory cluster on $c$-packed trajectories. The lemma uses the $c$-packedness property in two different ways. First, the $c$-packedness property bounds the complexity of the simplified free space diagram to linear. This replaces one of the factors of~$n$ with $c/\varepsilon$. Second, the $c$-packedness property is used to prove that in any horizontal strip, only a constant number of cells have free space. This replaces another factor of~$n$ with $c/\varepsilon$, resulting in~$O(c^2n/\varepsilon^2)$ internal critical points. Finally, we prove that the interval management data structure can be used in the same way as in Gudmundsson and Wong's algorithm~\cite{GudmundssonWong2022Cubicupperlower}. Putting this all together, we obtain Theorem~\ref{thm:arbitrary_reference}. For a full proof, see Section~\ref{sec:arbitrary_ref_trajectory}.

\begin{restatable}{theorem}{arbitraryreference}
\label{thm:arbitrary_reference}
There is an $O((c^2 n/\varepsilon^2)\log(c/\varepsilon)\log(n/\varepsilon))$ time algorithm that solves $\mathtt{SC(T, m, l, (1 + \varepsilon)d)}$ in the case that the reference trajectory is arbitrary.
\end{restatable}

\section{Computing the Free Space Diagram} \label{sec:free_space_diagram}
In this section, we will explain the process of constructing a simplified free space diagram for two $c$-packed polygonal curves $P$ and $Q$. The free space $\mathcal{D}_d(P, Q)$ describes all pairs of points, one on $P$, one on $Q$, whose distance is at most $d$ \cite{altComputingFrechetDistance1995}. With slight abuse of notation, we parameterise the polygonal curve $P$ such that $P(x)$ is a point on $P$, where $x \in [0, \abse{P}]$. Formally, 
\begin{align*}
    \mathcal{D}_d(P, Q) = \{(x, y) \in [0, \abse{P}] \times [0, \abse{Q}] \mid dist(P(x), Q(y)) \leq d\}.
\end{align*}

To circumvent the quadratic free space complexity, Driemel et al.~\cite{DriemelEtAl2012ApproximatingFrechetDistance} showed that the free space complexity of two simplified $c$-packed curves is $O(cn/\e)$. Given a $c$-packed curve $P = p_1p_2...p_n$, we simplify $P$ into its $\ed$-simplification $P' = q_1q_2...q_k$ as follows. Let $B(a, r)$ be the ball centered at $a$ with radius $r$. First, set $q_1 = p_1$. With $q_i$ defined, traverse $P$ from $q_i$ until a vertex $v$ is outside $B(q_i, \varepsilon d)$, or $v$ is the last vertex of $P$, and set $q_{i+1} = v$. Continue until all vertices of $P$ are exhausted. Driemel et al.~\cite{DriemelEtAl2012ApproximatingFrechetDistance} showed that the $\ed$-simplification of a $c$-packed curve is at most $6c$-packed~\cite[Lemma~4.3]{DriemelEtAl2012ApproximatingFrechetDistance}, and the Fréchet distance between $P$ and $P'$ is at most $(1 + \e)d$. A simplified curve has the useful property that every segment but the last is at least $\ed$ long. We assume for simplicity that the last one is at least $\ed$ long as well, since otherwise one can backtrack and modify the simplified curve such that each segment is at least $\ed/2$ long, and our arguments can be extended to such case. 

\begin{observation}
    One can simplify a polygonal curve $P$ into its $\ed$-simplification $P'$ such that the \Frechet distance between $P$ and $P'$ is at most $(1 + \e)d$, and each segment in $P'$ is at least $\ed$ long.
\end{observation}

Simplifying two $c$-packed curves can reduce the free space complexity, but using the plane-sweep algorithm to solve the SC problem on the resulting free space diagram is unfortunately infeasible. This is because the total length of the simplified trajectories can be much shorter, making it impossible to slide a window of width $l$ on the free space diagram $\mathcal{F}_{(1 + \e)d}(P', Q')$. To address this issue, we developed a tool that enables the construction of a free space diagram that maintains the original curve length, while also benefiting from the reduced free space complexity.

\subsection{Simplifying the Free Space}
In this section, we introduce a method that simplifies the free space. We show that we can construct the \textit{simplified free space} $\mathcal{D}'_{(1 + \hat{\e})d}(P, Q)$, where $\ehat$ is at most $8\e$, such that the complexity of the simplified free space is at most $O(cn/\ehat)$. In addition, $\mathcal{D}'_{(1 + \ehat)d}(P, Q)$ contains $\mathcal{D}_d(P, Q)$ as a subset, but it is not bigger than the free space of $P$ and $Q$ if we approximate their \Frechet distance, i.e., $\mathcal{D}'_{(1 + \ehat)d}(P, Q) \subseteq \mathcal{D}_{(1 + \ehat)d}(P, Q)$.

We will first define a function that uniformly maps parts of the polygonal curve $P$ to segments of $P'$ in Definition~\ref{def:function}, using which we will formally define the simplified free space in Definition~\ref{def:simpl-free-space}. We will then formally prove the set inclusions mentioned above in Lemma~\ref{lemma:free_space_inclusive}. 

\begin{definition}
\label{def:function}
Let $simpl(P, \ed)$ be the $\ed$-simplification of a polygonal curve $P$. Let $P_{uv}$ be the subcurve of $P$ from point $u$ to $v$ that are simplified into the segment $(u, v) \in simpl(P, \ed)$. Let $w$ be the first intersection point of $P_{uv}$ and the boundary of the ball $B(u, \ed)$ along $P_{uv}$, and let $u'$ be the intersection of $(u, v)$ and the boundary of the ball $B(u, \ed)$. Define the mapping $f_{P, \ed}: P \rightarrow simpl(P, \ed)$ such that $f_{P, \ed}$ maps $[u, w)$ to $[u, u')$ uniformly, and $[w, v]$ to $[u', v]$ uniformly (see Figure~\ref{fig:uv_simpl}).
\end{definition}

\begin{figure}[tbh]
    \centering
    \includegraphics[scale=0.75]{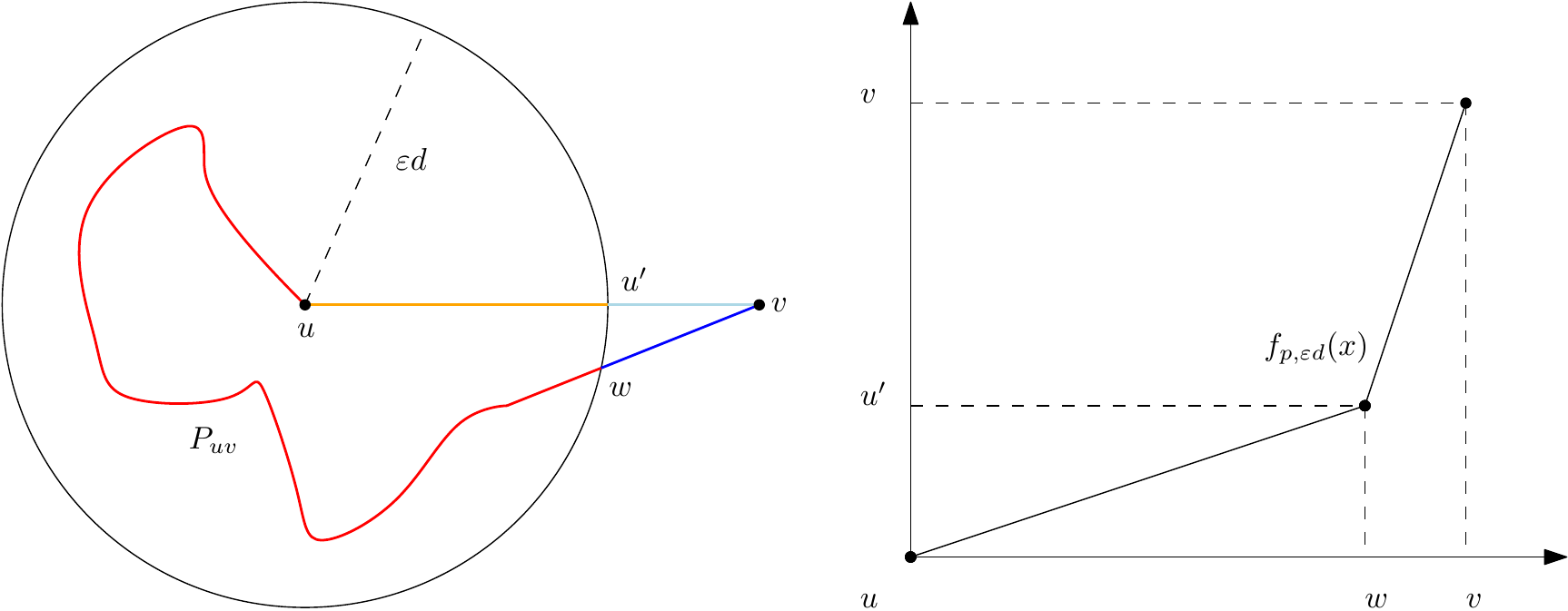}
    \caption{A figure showcasing the function in Definition~\ref{def:function}. The point $u'$ is the intersection of the segment $(u, v)$ and the ball $B(u, \ed)$, and the point $w$ is the intersection of subtrajectory $P_{uv}$ and $B(u, \ed)$. The function $f_{P, \ed}$ uniformly maps $P_{uw}$ (red) to $(u, u')$ (orange), not including $u'$ and $w$. The function $f_{P, \ed}$ uniformly maps $P_{wv}$ (blue) to $(u', v)$ (light blue).}
    \label{fig:uv_simpl}
\end{figure}

%Joachim: Moved for conference version to bottom of page.
% \begin{figure}[tbh]
%    \centering
%    \includegraphics[scale=0.75]{images/uv_simpl.pdf}
%    \caption{A figure showcasing the function in Definition~\ref{def:function}. The function $f_{P, \ed}$ uniformly maps $P_{uw}$ to $(u, u')$ and $P_{wv}$ to $(u', v)$.}
%    \label{fig:uv_simpl}
% \end{figure}

\begin{definition}
\label{def:simpl-free-space}
Define the simplified free space of $P$ and $Q$ with respect to the \Frechet distance $d > 0$, and a parameter $\e > 0$ as
\begin{align*}
    \mathcal{D}'_{(1 + \e)d}(P, Q) = \{(x, y) \in [0, \abse{P}] \times [0, \abse{Q}] \mid dist(f_{P, \ed}(P(x)), f_{Q, \ed}(Q(y))) \leq (1 + \e) d\}.
\end{align*}
Similarly, let $\mathcal{F}'_{(1 + \e)d}(P, Q)$ be the simplified free space diagram.
\end{definition}

\begin{lemma}
\label{lemma:free_space_inclusive}
    Let $\mathcal{D}_d(P, Q)$ be the free space of curves $P$ and $Q$ with respect to the \Frechet distance $d$, and let $\mathcal{D}'_{(1 + \e)d}(P, Q)$ be their simplified free space with an approximation error $\e > 0$. Then
    $\mathcal{D}_d(P, Q) \subseteq \mathcal{D}'_{(1 + 4\e) d}(P, Q) \subseteq \mathcal{D}_{(1 + 8\e)d} (P, Q)$.
\end{lemma}

\begin{proof}
    With slight abuse of notation, let $x = P(x)$, and $y = Q(y)$, for $x \in [0, \abse{P}]$, and $y \in [0, \abse{Q}]$. Let $x' = f_{P, \ed}(x)$, and let $y' = f_{Q, \ed}(y)$. Observe that $dist(x, x') \leq 2\varepsilon d$ for all $x \in P$ because if $x$ is within the ball $B(u, \varepsilon d)$, then $x$ is at most $2 \varepsilon d$ apart from $x'$. If $x$ is outside $B(u, \ed)$, it is at most $\varepsilon d$ apart from $x'$, due to the simplification.  
    \begin{itemize}
        \item $\mathcal{D}_d(P, Q) \subseteq \mathcal{D}'_{(1 + 4\varepsilon) d}(P, Q)$. If a point $(x, y) \in \mathcal{D}_d(P, Q)$ is white, then $dist(x, y) \leq d$. By the triangle inequality, $dist(x', y') \leq dist(x', x) + dist(y', y) + dist(x, y) \leq 2 \varepsilon d + 2 \varepsilon d + d = (1 + 4\varepsilon) d$, hence $(x', y')$ must also be white. 
        \item $ \mathcal{D}'_{(1 + 4\varepsilon) d}(P, Q) \subseteq \mathcal{D}_{(1 + 8\varepsilon)d} (P, Q)$. Similarly, if a point $(x', y') \in \mathcal{D}'_{(1 + 4\e) d}(P, Q)$ is white, then $(x, y) \in \mathcal{D}_{(1 + 8\varepsilon)d} (P, Q)$ must also be white, because $dist(x, y) \leq dist(x', x) + dist(y', y) + dist(x', y') \leq 2 \varepsilon d + 2 \varepsilon d + (1 + 4\varepsilon)d = (1 + 8\varepsilon) d$.
        \qedhere
    \end{itemize}
\end{proof}

Similar to how we defined the $(u, v)$ cell, let the $(P_{uv}, Q_{ab})$ cells be the cells in the free space diagram defined by the subcurves $P_{uv}$ and $Q_{ab}$. We show that we can compute the intersection of the simplified free space with $(P_{uv}, Q_{ab})$ cells in constant time.

\begin{restatable}{lemma}{cellsconstant}
\label{lemma:cells_constant}
    Given vertices $u, v$ on $P$ and $a, b$ on $Q$, one can construct the cells in $\mathcal{F}'_{(1 + \e)d}(P, Q)$ defined by $P_{uv}$ and $Q_{ab}$ in constant time. 
    % Let $P_{uv}$ be the set of segments starting from $u$ and ending at $v$ that is simplified into $(u, v) \in simpl(P, \ed)$ segment, and similarly, $Q_{ab}$ is simplified into $(a, b)$. Construction of the cell defined by $P_{uv}$ and $Q_{ab}$ in the simplified free space diagram $\mathcal{F}'_{(1 + \e)d}(P, Q)$ takes constant time. 
\end{restatable}

\begin{proof}
See Appendix~\ref{appendix:cells_constant}.
\end{proof}

The complexity of the simplified free space $\mathcal{D}'_{(1 + \ehat)d}(P, Q)$ is $O(cn/\ehat)$ if $P$ and $Q$ are $c$-packed. Assuming that $P_{uv}$ and $Q_{ab}$ are simplified into segments $(u, v) \in P'$ and $(a, b) \in Q'$, respectively, the simplified free space intersects ($P_{uv}, Q_{ab}$) cells if and only if the distance between $(u, v)$ and $(a, b)$ is at most $(1 + \ehat)d$. The rest follows by modifying the proof of \cite[Lemma~4.4]{DriemelEtAl2012ApproximatingFrechetDistance}.

\begin{corollary} \label{cor:fewer_cells}
    Let $P$ and $Q$ be two $c$-packed curves with complexity $n$, and let $\ehat$ be a constant times a parameter $\e > 0$. The complexity of the simplified free space $\mathcal{D}'_{(1 + \ehat)d}(P, Q)$ is $O(cn/\e)$. 
\end{corollary}

\subsection{Compute the Non-empty Cells} \label{ssec:compute_nonempty_cells}
% Given the $\ed$-simplifications of two $c$-packed curves, \citet{DriemelEtAl2012ApproximatingFrechetDistance} showed that, in the free space diagram, there are at most $O(cn/\varepsilon)$ non-empty cells. 
To take advantage of the near-linear complexity of the simplified free space, we use an algorithm by Conradi and Driemel~\cite{conradiComputingKShortcutFrechet2022} to efficiently compute the non-empty cells  without inspecting all pairs of segments. 

\begin{fact}[{\cite[Lemma 59]{conradiComputingKShortcutFrechet2022}}] \label{fact:conradiDriemel}
    Given two $c$-packed curves $P$ and $Q$ in $\mathbb{R}^2$, a parameter $d \geq 0$, and let $P'$ and $Q'$ be their $\ed$-simplifications. In $O((cn/\e) \log(cn/\e))$ time, one can find all pairs of segments $(u, v) \in P'$ and $(a, b) \in Q'$ such that the distance between $(u, v)$ and $(a, b)$ is at most $d$.
\end{fact}

To construct the simplified free space diagram efficiently, we first observe the following.

\begin{observation}
\label{obs:simplified_segments_near}
    If segments $(u, v) \in P'$ and $(a, b) \in Q'$ are more than $(1 + 2\e)d$ apart, then $P_{uv}$ and $Q_{ab}$ are more than $d$ apart. 
    % Otherwise, $P_{uv}$ and $Q_{ab}$ may or may not be at most $d$ apart.
\end{observation}

The above observation enables us to determine if $(P_{uv}, Q_{ab})$ cells are empty by determining if $(u, v)$ and $(a, b)$ are near. 

\subsection{Constructing the Simplified Free Space Diagram}
Given two $c$-packed polygonal curves $P$ and $Q$, we will use the results from previous subsections to construct the simplified free space diagram using the below steps. In Lemma~\ref{lemma:cells_constant}, we showed that if $P_{uv}$ and $Q_{ab}$ are simplified into segments $(u, v) \in P'$ and $(a, b) \in Q'$, respectively, we can compute $(P_{uv}, Q_{ab})$ cells in constant time. Such aggregation of $(P_{uv}, Q_{ab})$ cells is an \textit{aggregated non-empty cell}, and we will treat them as one cell for simplicity.

\begin{enumerate}
    \item Simplify $P$ and $Q$ into their $\ed$-simplifications $P'$ and $Q'$. 
    \item Find all pairs of nearby segments from $P'$ and $Q'$ that are at most $(1 + \ehat)d$ apart using Fact~\ref{fact:conradiDriemel}.
    \item For each pair of nearby segments $(u, v) \in P'$ and $(a, b) \in Q'$, compute the $(P_{uv}, Q_{ab})$ cells using Lemma~\ref{lemma:cells_constant}.
    \item Sort all non-empty cells horizontally and vertically.
    \item Connect non-empty cells in a graph fashion such that a non-empty cell is connected to the first non-empty cells to its top, bottom, left, and right.

\end{enumerate}

\begin{figure}[tbh]
   \centering
   \includegraphics[scale=0.4]{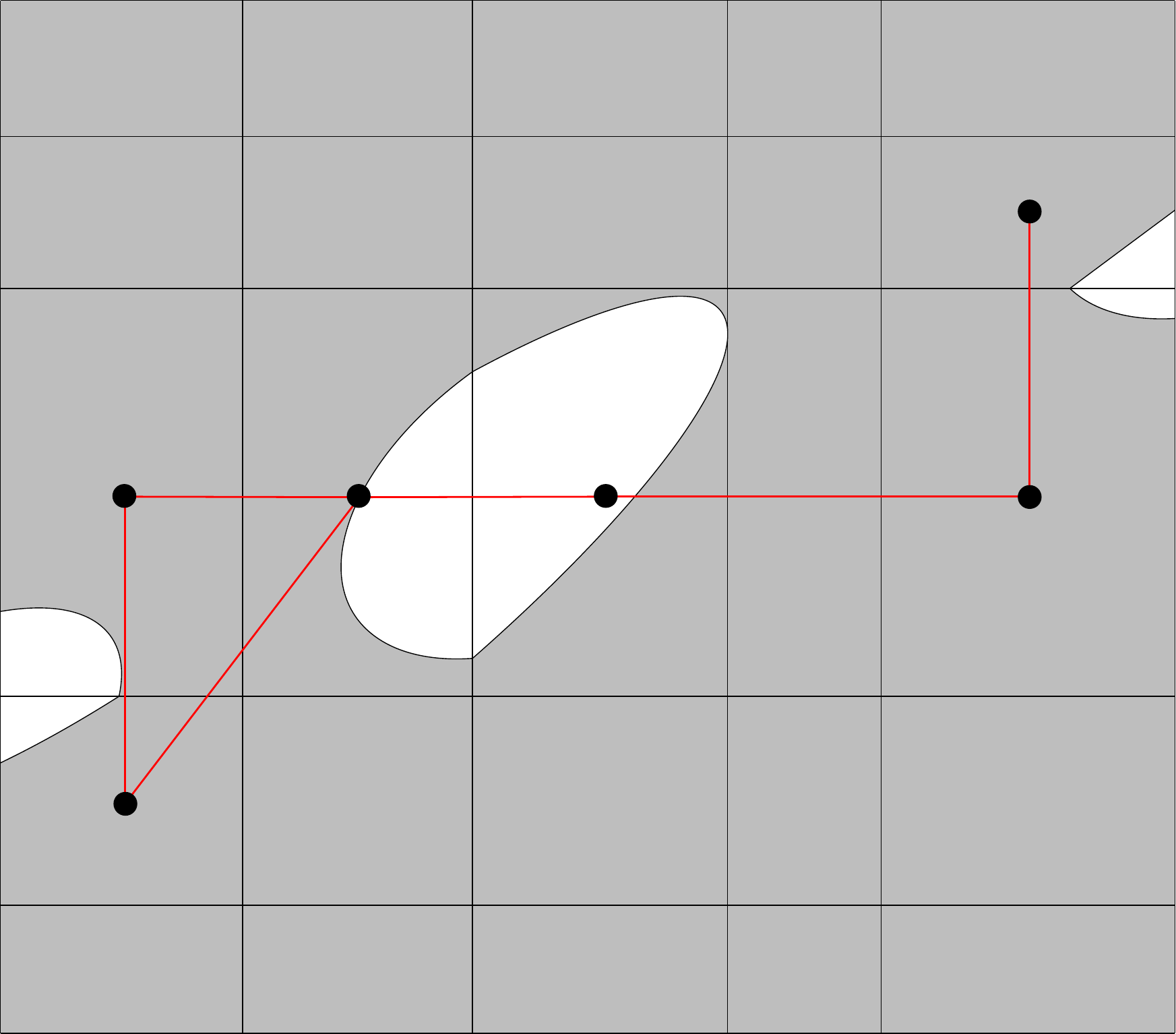}
   \caption{The non-empty cells are connected horizontally and vertically to skip empty cells.}
   \label{fig:free_cells_connected_using_linked_list}
\end{figure} %sw

Given two polygonal curves $P$ and $Q$ of complexity $n$, simplifying them (step 1) takes $O(n)$ time. By Fact~\ref{fact:conradiDriemel}, step 2 takes $O((cn/\varepsilon)\log(cn/\varepsilon))$ time. Computing a cell in $\mathcal{F}'_{(1 + \ehat)d}(P, Q)$ takes $O(1)$ time by Lemma~\ref{lemma:cells_constant}. $\mathcal{F}'_{(1 + \ehat)d}(P, Q)$ has at most $O(cn/\e)$ non-empty cells, which takes $O(cn/\e)$ time to compute in step 3; sorting them in step 4 takes $O((cn/\e)\log(cn/\e))$ time. Connecting each cell to at most four other cells takes $O(cn/\e)$ time in step 5. Putting this together, we obtain Lemma~\ref{lem:free_space_diagram_construction}, and we summarise our result in Theorem~\ref{thm:free_space_summarised}.

\begin{lemma}
\label{lem:free_space_diagram_construction}
Let $P$ and $Q$ be two $c$-packed curves of complexities $n$. Let $\e > 0$ and $d > 0$ be two parameters, and let $\ehat \leq 8 \cdot \e$. One can construct and connect $O(cn/\e)$ aggregated non-empty cells of the simplified free space diagram $\mathcal{F}'_{(1 + \ehat)d}(P, Q)$ in $O(\constructFdComplexity)$ time such that $\mathcal{D}_d(P, Q) \subseteq \mathcal{D}'_{(1 + \ehat) d}(P, Q) \subseteq \mathcal{D}_{(1 + \ehat)d} (P, Q)$. Given an aggregated non-empty cell $C$, one can access the first aggregated non-empty cells below, above, to the left, and to the right of $C$ in $O(1)$ time. 
\end{lemma}

\freespacesummarised*

\section{Reference trajectory is vertex-to-vertex} \label{sec:vertex_to_vertex}

Throughout the rest of the paper we assume that the free space diagram is the simplified free space diagram $\mathcal{F}'_{(1 + \e)d}$ in Lemma~\ref{lem:free_space_diagram_construction}. Next, we will use the algorithm by Gudmundsson and Wong \cite{GudmundssonWong2022Cubicupperlower} to determine whether there is a solution to $\mathtt{SC(T, m, l, (1 + \varepsilon)d)}$ where $T$ is a $c$-packed trajectory, and the reference subtrajectory is vertex-to-vertex.

Three data structures are used in the vertex-to-vertex subtrajectory cluster algorithm of Gudmundsson and Wong~\cite{GudmundssonWong2022Cubicupperlower} --- a directed graph, a range tree, and a link-cut tree. For an overview of these data structures, see Appendix~\ref{appendix:graph_construction}. In Section~\ref{sec:graph_construction_new}, we show that the number of leaves per range tree is~$O(c/\varepsilon)$, and the directed graph has complexity~$O((cn/\varepsilon) \log (c/\varepsilon))$. In Section~\ref{sec:linkcut_tree}, we show that the link-cut tree data structure can be used without modification.

\subsection{Using a Directed Graph to Store Candidate Monotone Paths}
\label{sec:graph_construction_new}

To show that the range tree has at most~$O(c/\varepsilon)$ leaves, it suffices to show that there exist at most $O(c/\varepsilon)$ critical points on each horizontal or vertical boundary of the simplified free space diagram. 

\begin{lemma}
\label{lemma:strip_limit}
In the simplified free space diagram $\mathcal{F}'_{(1 + \ehat)d}(T, T)$, let $H$ be a horizontal (resp. vertical) strip that is at least $\ed$ wide on its $y$-span (resp. $x$-span). The intersection of $H$ and the simplified free space $\mathcal{D}'_{(1 + \ehat)d}(T, T)$ exists in at most $O(c/\e)$ aggregated cells.
\end{lemma}

\begin{proof}
Let $T'$ be the $\ed$-simplification of $T$, and let $T_{uv}$ simplifies into segment $(u, v) \in T'$. Let $u' \subseteq (u, v)$ be a small part that is at least $\ed$ long. Let $S_{u'} = u' \oplus B(0, (1 + \ehat)d)$. 
  
Using similar construction, and arguments of \cite[Lemma 4.4]{DriemelEtAl2012ApproximatingFrechetDistance}, one can prove that at most $O(c/\e)$ segments in $T'$ intersects $S_{u'}$. Based on the construction of the simplified free space $\mathcal{D}'_{(1 + \ehat)d}(T, T)$, a point $(x, y) \in \mathcal{D}'_{(1 + \ehat)d}(T, T)$ is white if and only if $dist(f_{T, \ehat d}(T(x)), f_{T, \ehat d}(T(y))) \leq (1 + \ehat)d$. As such, at most $O(c/\e)$ aggregated cells have simplified free space intersecting $H$.
% Let $B$ be a ball centered at the midpoint of $u'$ with radius $(3\varepsilon/2 + 1)d$, and recall that $S_{u'}$ is the stadium of radius $d$ around $u'$. Observe that the minimum distance from the border of $P_{u'}$ to the border of $B$ is $\ed$. Therefore the total number of segments intersecting $S_{u'}$ is:
% \begin{align*}
%     \text{\#segments} &\leq \frac{\text{maximum length of } \abse{T' \cap B}}{\text{minimum length of segments covered by $B$}} \\
%     &= \frac{6c \cdot (3\varepsilon/2 + 1)d}{\ed} && \text{($T'$ is $6c$-packed \cite{DriemelEtAl2012ApproximatingFrechetDistance})}\\
%     &= 9c + \frac{6c}{\varepsilon} \\
%     &\in O\left(\frac{c}{\varepsilon}\right)
% \end{align*}
% In order for a $(u', v)$ cell to have free space intersecting the vertical strip $u' \times [0, \infty]$ or the horizontal strip $[0, \infty] \times u'$, $v$ must intersect $S_{u'}$. Since we have shown that only $O(c/\varepsilon)$ segments can intersect $B$ and thus $S_{u'}$, the proof is complete.
\end{proof}

% In summary, we showed that for a critical point $p_i$ on the vertical boundaries of a single row, we can find the rightmost point $q_i$ such that there is a $p_i q_i$ basic monotone path in $O(\log{n_k})$ time in Lemma~\ref{lemma:find_qi}. And, we showed that $p_i$ connects to at most $O(c/\varepsilon)$ nodes. %sw

Next, bound the construction time and space complexity of the directed graph in~\cite{GudmundssonWong2022Cubicupperlower}. 

% \citeauthor{DriemelEtAl2012ApproximatingFrechetDistance} showed that the number of non-empty cells in the free space diagram $\mathcal{F}_d(T, T)$ is at most $O(cn/\varepsilon)$ if $T$ is the $\ed$-simplification of a $c$-packed curve \cite[Lemma~4.4]{DriemelEtAl2012ApproximatingFrechetDistance}. We summarise the result in the below lemma.

\begin{lemma}
\label{lemma:constructing_G}
Given a $c$-packed trajectory $T$ of complexity $n$, constructing $G$ for the simplified free space diagram $\mathcal{F}'_{(1 + \ehat)d}(T, T)$ takes $O((cn/\varepsilon) \log{(n/\varepsilon)})$ time. $G$ has $O(cn/\varepsilon)$ nodes and $O((cn/\varepsilon) \log{(c/\varepsilon)})$ edges.
\end{lemma}

\begin{proof}
Let $n_k$ be the number of non-empty aggregated cells in the $j$th row in $\mathcal{F}'_{(1 + \ehat)d}(T, T)$. Construction of the range tree for the top (resp. right) boundary of a row (resp. column) takes $O(n_k \log{n_k})$ time~\cite{debergComputationalGeometryAlgorithms2008}. For all $p_i$, finding $q_i$ takes $O(n_k \log{n_k})$ time
% by Lemma~\ref{lemma:find_qi}, %sw
and recall that there are $O(cn/\varepsilon)$ critical points in $\mathcal{F}'_{(1 + \ehat)d}(T, T)$. The total construction time is as follows.

\begin{align*}
     \sum_{j=0}^{n + 1} n_k \log{n_k} \leq 
        \log{\left(\frac{cn}{\varepsilon}\right)} \sum_{j=0}^{n + 1} n_k = 
         \log{\left(\frac{cn}{\varepsilon}\right)} O\left(\frac{cn}{\varepsilon}\right) \in 
        O\left(\left(\frac{cn}{\varepsilon}\right) \log{\left(\frac{n}{\varepsilon}\right)}\right)
\end{align*}

By Corollary~\ref{cor:fewer_cells}, the simplified free space diagram has $O(cn/\varepsilon)$ non-empty aggregated cells, therefore $G$ has $O(cn/\varepsilon)$ nodes. In a range tree, given a continuous interval $[q_k, q_i]$, one can find $O(\log{n})$ nodes such that these nodes include $[q_k, q_i]$ in their canonical subset, where $n$ is the total number of items in the leaves \cite{debergComputationalGeometryAlgorithms2008}. There are at most $O(c/\varepsilon)$ nodes on a horizontal or vertical boundary by Lemma~\ref{lemma:strip_limit}, and each critical point $p_i$ on a vertical (resp. horizontal) cell boundary connects to $O(\log(c/\varepsilon))$ nodes, therefore the total number of edges is $O((cn/\varepsilon) \log(c/\varepsilon))$.
\end{proof}

\subsection{Storing and Reusing Pre-computed Paths}
\label{sec:linkcut_tree}
A link-cut tree \cite{sleatorDataStructureDynamic1983} maintains a forest that allows the link/cut operations of subtrees in $O(\log{n})$ amortised time. In addition, a link-cut tree allows finding the root of a node in $O(\log{n})$ amortised time. The algorithm by Gudmundsson and Wong \cite{GudmundssonWong2022Cubicupperlower} used a link-cut tree to store and re-use monotone paths. Consider when a sweepline, either $l_s$ or $l_t$, stops at a new critical point $p$. Instead of recomputing the monotone paths, they need only to add $p$ to the existing link-cut tree they maintained in the previous instances. 

With graph $G$ defined, we can analyse the total running time of the algorithm by Gudmundsson and Wong \cite{GudmundssonWong2022Cubicupperlower} on the simplified free space diagram. The key to observe the running time is that in their algorithm, if an edge leads to a dead-end, it is marked and will not be used in future searches. Furthermore, inserting or removing an edge takes $O(\log n)$ amortised time in a link-cut tree.

\vertextovertex*

\begin{proof}
\sloppy Construction of the simplified free space diagram takes $O(\constructFdComplexity)$ time by Theorem~\ref{thm:free_space_summarised}. Construction of $G$ takes $O((cn/\varepsilon)\log{(n/\varepsilon)})$ time by Lemma~\ref{lemma:constructing_G}. The graph~$G$ has at most $O(nm\log(c/\varepsilon))$ edges, see Appendix~\ref{appendix:greedy_crit_points}. Gudmundsson and Wong \cite{GudmundssonWong2022Cubicupperlower} showed that an edge is added to and removed from the link-cut tree at most once, and adding/removing an edge from the link-cut tree takes $O(\log(n/\varepsilon))$ time since the maximum number of nodes in the link-cut tree is upperbounded by the number of nodes in $G$. Therefore maintaining the link-cut tree takes $O(nm\log(c/\varepsilon)\log(n/\varepsilon))$ time.
\end{proof}

\section{Reference trajectory is arbitrary}
\label{sec:arbitrary_ref_trajectory}
Our results in this section rely heavily on the work of Gudmundsson and Wong~\cite{GudmundssonWong2022Cubicupperlower}. Due to space constraints, we can only highlight important parts of their algorithm, and the analysis of our improvements. 

When the reference trajectory is arbitrary, a monotone path can start and finish at arbitrary positions in the non-empty cells. Therefore, in addition to the critical points in the free space diagram and the greedy critical points, Gudmundsson and Wong defined three new types of \textit{internal critical points} \cite[Definition~25]{GudmundssonWong2022Cubicupperlower}. An internal critical point must lie in the interior of a non-empty cell, and lie on the boundary of the free space. They made the following distinction (see Figure~\ref{fig:internal_critical_points}).

\begin{enumerate}
    \item End-of-cell critical point: the leftmost and rightmost white point of a non-empty cell.
    \item Propagated critical point: a point on the boundary of the free space that shares a $y$-coordinate with a critical point.
    \item $l$-apart critical points: two points on the boundaries of free space that are a distance of $l$ apart horizontally.
\end{enumerate}

\begin{figure}[bth]
    \centering
    \includegraphics[width=\textwidth]{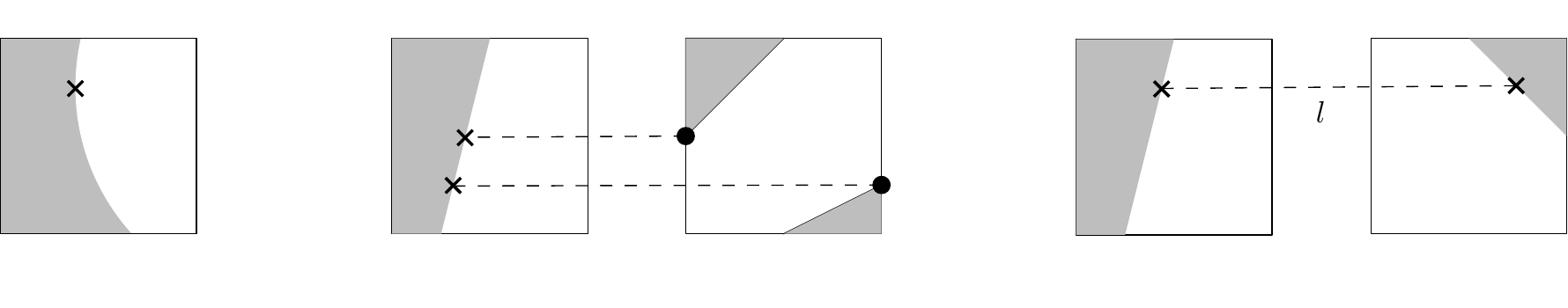}
    \caption{The three types of internal critical points are illustrated using a cross in the left, middle, and right figures, respectively. From left to right, they are the end-of-cell critical points (left), propagated critical points (middle) and $l$-apart critical points (right).}
    \label{fig:internal_critical_points}
\end{figure}

%It is worth noting that 
There could be an infinite number of $l$-apart critical points in a pair of non-empty cells. However, if this is the case, we can simply perturb the input by a miniscule amount so that there are no longer an infinite number of  $l$-apart critical points. See Appendix~\ref{appendix:l_apart} for an example and a figure. Therefore, for the rest of the paper, we can assume that there are at most a constant number of $l$-apart critical points per pair of cells.

% \begin{remark}
%     One can detect that two cells generate infinite $l$-apart critical points by comparing the quadratic equations of their interiors. To cope with this case, one can move the end point of one of the segment by a miniscule amount, say $\delta = 1/\varepsilon^2$.
% \end{remark} %sw

We will first bound the number of internal critical points and the time it takes to compute them. One can compute the end-of-cell and $l$-apart critical points in linear time with respect to the number of non-empty cells since there are at most a constant number of them per pair of cells. In Lemma~\ref{lemma:strip_limit}, we showed that in a narrow horizontal strip, only a small number of cells intersect free space. An output-sensitive query algorithm would be efficient to find the non-empty cells that a critical point $p$ propagates to. Therefore, we can use an interval tree~\cite{debergComputationalGeometryAlgorithms2008} to store the $y$-spans of all non-empty cells in a row, and query the intersecting intervals of $y(p)$ in logarithmic time. We formalise the above arguments in the below Lemma~\ref{lemma:number_of_internal_crit_points}.

\begin{lemma}
\label{lemma:number_of_internal_crit_points}
Assume that there is a constant number of $l$-apart critical points per pair of cells, it takes $O(cn\log(n/\varepsilon) + c^2n/\varepsilon^2)$ time to compute $O(c^2n/\varepsilon^2)$ internal critical points in the simplified free space diagram $\mathcal{F}'_{(1 + \ehat)d}(T, T)$.
\end{lemma}

\begin{proof}
There are $O(cn/\varepsilon)$ non-empty aggregated cells in $\mathcal{F}'_{(1 + \ehat)d}(T, T)$, or non-empty cells for short, and $O(cn/\varepsilon)$ end-of-cell critical points in $\mathcal{F}'_{(1 + \ehat)d}(T, T)$. Each critical point propagates $O(c/\varepsilon)$ times by Lemma~\ref{lemma:strip_limit}, therefore there are $O(c^2 n/\varepsilon^2)$ propagated critical points. We can charge a cell with a constant number of $l$-apart critical points. Therefore, there are at most $O(cn/\e)$ $l$-apart critical points. In total, there are $O(c^2 n/\varepsilon^2)$ internal critical points.

One can compute the end-of-cell critical points by iterating through the free space diagram in $O(cn/\varepsilon)$ time. To compute the $l$-apart critical points, we can start from the first non-empty cell $C$ in a row and find the first cell that is $l$-apart from $C$, and solve a constant number of quadratic equations. We can then slide this $l$-apart line and do the same for all pairs of cells that are $l$-apart in all rows in $O(cn/\varepsilon)$ time in total. 

To compute the propagated critical points, we construct an interval tree \cite{debergComputationalGeometryAlgorithms2008} for each row in $\mathcal{F}'_{(1 + \ehat)d}(T, T)$ to store the maximum and minimum $y$-coordinates of the free space in the non-empty cells. Let $n_i$ be the number of non-empty cells in the $i$th row. We can sum the construction time of the interval trees.
\begin{align*}
    \sum_{i = 1}^{n} n_i \log n_i \leq \frac{cn}{\varepsilon} \log \left( \frac{cn}{\varepsilon} \right) \in O\left( \left( \frac{cn}{\varepsilon} \right) \log \left( \frac{cn}{\varepsilon} \right) \right) 
\end{align*}

Given a critical point $p$ in the $i$th row, one can query the interval tree in $O(\log{n_i} + c/\varepsilon) \in O(\log n + c/\varepsilon)$ time to compute the propagated critical points from $p$ using Lemma~\ref{lemma:strip_limit} and \cite{debergComputationalGeometryAlgorithms2008}. With $O(cn/\varepsilon)$ critical points, computing the propagated critical points takes $O(cn\log(n)/\varepsilon + c^2n/\varepsilon^2)$ time.
\end{proof}

With the additional internal critical points, the number of reference trajectories and the number of greedy critical points increases. We can use the algorithm in the previous section, and obtain the following result.

% without the interval management data structure
\begin{restatable}{lemma}{arbitraryRefNoIntervalManagement}
\label{lemma:arbitrary_reference_no_interval_management}
\sloppy There is an $O((c^2m n/\varepsilon^2)\log(c/\varepsilon)\log(n/\varepsilon))$ time algorithm that solves $\mathtt{SC(T, m, l, (1 + \varepsilon)d)}$ in the case that the reference trajectory is arbitrary.
\end{restatable}

\begin{proof}
See Appendix~\ref{appendix:arbitrary_reference_no_interval_management}.
\end{proof}

\subsection{Improve Further with an Interval Management Data Structure}
The bottleneck in the above Lemma~\ref{lemma:arbitrary_reference_no_interval_management} is operating the outgoing edges of the $O(c^2mn/\varepsilon^2)$ greedy critical points, which are generated from $O(c^2n/\varepsilon^2)$ propagated critical points. To avoid computing the greedy critical points, Gudmundsson and Wong \cite{GudmundssonWong2022Cubicupperlower} used a dynamic monotonic interval data structure \cite{GavruskinEtAl2015Dynamicalgorithmsmonotonic} to store overlapping monotonic intervals that represent the $y$-spans of monotone paths between $l_s$ and $l_t$. Instead of searching for a set of monotone paths between each window greedily, they showed that one can update and query the interval data structure to retrieve $m - 1$ non-overlapping intervals, all in $O(\log n)$ amortised time. 

\arbitraryreference*

\begin{proof}
Constructing the simplified free space diagram takes $O(\constructFdComplexity)$ time by Theorem~\ref{thm:free_space_summarised}. Computing and sorting the internal critical points takes $O((c^2n/\varepsilon^2)\log(n/\varepsilon))$ time by Lemma~\ref{lemma:number_of_internal_crit_points}. There are $O((c^2n/\varepsilon^2)\log(c/\varepsilon))$ edges in total by Lemma~\ref{lemma:strip_limit}, and each edge takes $O(\log(n/\varepsilon))$ time to insert or remove since there are at most $O(c^2 n/\varepsilon^2)$ nodes in the link-cut tree. In total, we spend $O((c^2n/\varepsilon^2)\log(c/\varepsilon)\log(n/\varepsilon))$ time to maintain the edges in $G$. 

Each internal critical point is treated as an event, and maintaining the interval data structure takes $O(\log n)$ amortised time per event point (see \cite[Theorem~2]{GudmundssonWong2022Cubicupperlower}), and thus $O((c^2n/\varepsilon^2)\log n)$ in total. The overall complexity is dominated by maintaining the edges.
\end{proof}

\section{Conclusion}
We presented an algorithm that solves the subtrajectory cluster problem on $c$-packed trajectories $T$ with an approximation error on the Fréchet distance, achieving an \finalComplexity time complexity. Our algorithm builds upon the near-optimal algorithm proposed by Gudmundsson and Wong~\cite{GudmundssonWong2022Cubicupperlower}, but with significant improvements. By carefully analysing the properties of $c$-packed trajectories, we have shown that important parameters such as the number of propagated critical points are significantly lower than the theoretical $O(n)$ upperbound for realistic trajectories. As a result, our algorithm improves upon the near-optimal algorithm by replacing a factor of $n^2$ with $c^2/\varepsilon^2$, leading to more efficient subtrajectory cluster of realistic trajectories.
\newpage
\bibliographystyle{plain}
\bibliography{aaa}

% plainnat

\newpage

\appendix

\section[The Directed Graph of Gudmundsson and Wong]{The Directed Graph of Gudmundsson and Wong~\cite{GudmundssonWong2022Cubicupperlower}}
\label{appendix:graph_construction}
In this section, we will discuss the construction of a directed graph to store candidate monotone paths, as proposed by Gudmundsson and Wong~\cite{GudmundssonWong2022Cubicupperlower}. They defined a specific type of monotone path that exists in a row or column only, and showed that any general monotone path $P$ can be decomposed into a series of these so-called \textit{basic monotone paths} $P'$, such that the $y$-span of $P$ is a subset of the $y$-span of $P'$. The basic monotone paths are stored in a directed graph, which can be efficiently queried to find feasible monotone paths.

\begin{definition}[{\cite{GudmundssonWong2022Cubicupperlower}}]
A basic monotone path is a monotone path that is contained entirely in a single row or column of the free space diagram, starting at a critical point on a vertical cell boundary, and ending on a critical point on a horizontal cell boundary, or vice versa. 
\end{definition}

Gudmundsson and Wong~\cite[Lemma~16]{GudmundssonWong2022Cubicupperlower} showed that there is a monotone path from critical point $a$ to $b$ on the free space diagram if and only if there is a sequence of basic monotone paths between $a$ and $b$. Their idea is to decompose a monotone path into path $p_1p_2...p_k$ such that $p_1 = a$, $p_k = b$, and the path from $p_i$ to $p_{i + 1}$ is a basic monotone path. One can first transform a monotone path $Q$ into a set of almost-basic monotone paths $q_1...q_k$ inductively: if $q_i$ lies on a vertical (resp. horizontal) cell boundary, then $q_{i + 1}$ is the next intersection of $Q$ with a horizontal (resp. vertical) boundary. Then one can transform the path $q_1...q_k$ into a series of basic monotone path as follows. If $q_i$ lies on a vertical (resp. horizontal) cell boundary, then $p_i$ is the critical point below (resp. left of) $q_i$.

\begin{fact}[{\cite[Lemma~16]{GudmundssonWong2022Cubicupperlower}}]
Given a pair of critical points $a$ and $b$ in the free space diagram, there is an $ab$ monotone path if and only if there is a sequence of basic monotone paths $p_1...p_k$ such that $p_i$ is a critical point for $1 \leq i \leq k$, $p_1 = a$, and $p_k = b$.
\end{fact}

With basic monotone paths defined, we want to construct a graph to store all possible basic monotone paths. We will do so row-by-row and column-by-column. For an arbitrary row, where $H$ is the top boundary, let $p_i$ be the bottom-most critical point of the $i$th vertical cell boundary, where $0 \leq i \leq n + 1$. A brute-force approach is to connect every $p_i$ to every critical point $q$ on $H$ such that there is a $p_iq$ basic monotone path. However, the number of edges is cubic in this case. 

Define $q_i$ to be the rightmost critical point on the top boundary such that there is a $p_iq_i$ basic monotone path. A key observation is that if there is a $p_iq_i$ monotone path, then there is a $pq$ monotone path where $q$ is a critical point on $H$, $q$ is to the right of $p_i$ and to the left of $q_i$. Indeed, let $r$ be the intersection of the $p_iq_i$ monotone path with the left vertical cell boundary of the cell that contains $r$. Since the interior of a non-empty cell is convex, there is a $pr$ monotone path.

The above observation enables Gudmundsson and Wong to define a graph $G = (V, E)$ to efficiently store all possible basic monotone paths \cite{GudmundssonWong2022Cubicupperlower}. For an arbitrary row, they first construct a range tree $RT$ storing the critical points on $H$ with respect to their increasing $x$-coordinates. Then, they connect $p_i$ to the node $v$ in $RT$ as long as there is a basic monotone path from $p_i$ to every critical point in the leaves of $v$ (see Figure~\ref{fig:range_tree}). Each column of the free space diagram is processed analogously. As we will use the range tree extensively in the following section, we define the canonical subset to differentiate from the canonical squares defined in the previous section. Given a node $v$ in a range tree, \textit{the canonical subset} of $v$ is the set of points stored in the leaves of $v$ \cite{debergComputationalGeometryAlgorithms2008}.

% Such data structure contains $O(n^2 \log{n})$ edges. Because each range tree contains $O(n)$ edges per row or column, and each critical point is connected to at most $O(\log{n})$ nodes of a range tree. 

\begin{figure}[!htb]
    \centering
    \includegraphics[scale=0.8]{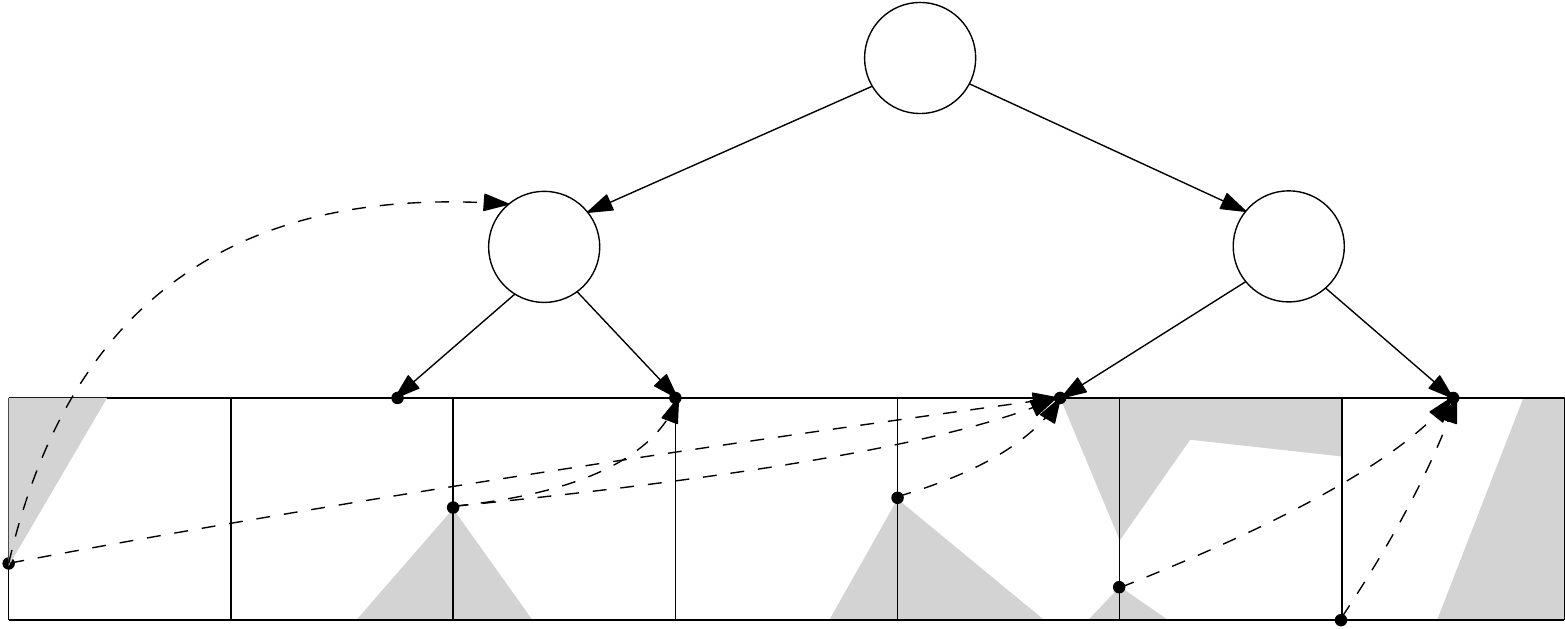}
    \caption{Part of graph $G$. We build a directed range tree from the critical points on the top boundary. A bottom critical point $p$ on a vertical cell boundary is connected to a node $v$ of the range tree as long as there is a basic monotone path from $p$ to every critical point in the canonical subset of $v$, which are illustrated with dashed lines.}
    \label{fig:range_tree}
\end{figure}

Gudmundsson and Wong~\cite{GudmundssonWong2022Cubicupperlower} described an algorithm that find $q_i$ for each $p_i$, However, we cannot use their algorithm for two reasons. First, their algorithm only works on the very restrictive case when all non-empty cells have critical points on their top boundary, which is not true in the free space diagram of two general curves. Second, their algorithm only works when all cells in a row are non-empty. 

We will now show how to find $q_i$ for each $p_i$. Let $r_i$ be the top critical point on a vertical cell boundary. For some $p_i$ on a $i$th vertical cell boundary, we can draw a horizontal line to the right from $p_i$. If this horizontal line is blocked first immediately before the the $j$th vertical cell boundary and $y(p_i) \leq y(p_j)$ (resp. $y(p_i) \geq y(r_j)$), we say $p_i$ is blocked by $p_j$ (resp. $r_j$).

There are two things that can happen (see Figure~\ref{fig:pi_blocked}). One, the line is blocked by some point $b$ on the boundary between free and non-free space in the cell immediately to the left of $r_j$ or $p_j$. $b$ is either higher than $r_j$ or lower than $p_j$. If $b$ is below some $p_j$, then $q_i = q_j$ since there is a $p_ip_jq_j$ monotone path. If $b$ is above some $r_j$, then $q_i$ is simply the rightmost critical point on $H$ that is to the left of $r_j$ and to the right of $p_i$. Second, the line from $p_i$ is blocked at point $b$ by the non-free space that completely separates two adjacent cells or it is blocked by the right boundary of the free space diagram. Say $b$ lies between the $j$th and the $(j-1)$th vertical cell, then $q_i$ is simply the rightmost critical point on $H$ that is to the left of the $j$th vertical cell boundary and to the right of $p_i$. 

\begin{figure}[!htb]
    \centering
    \includegraphics[scale=0.8]{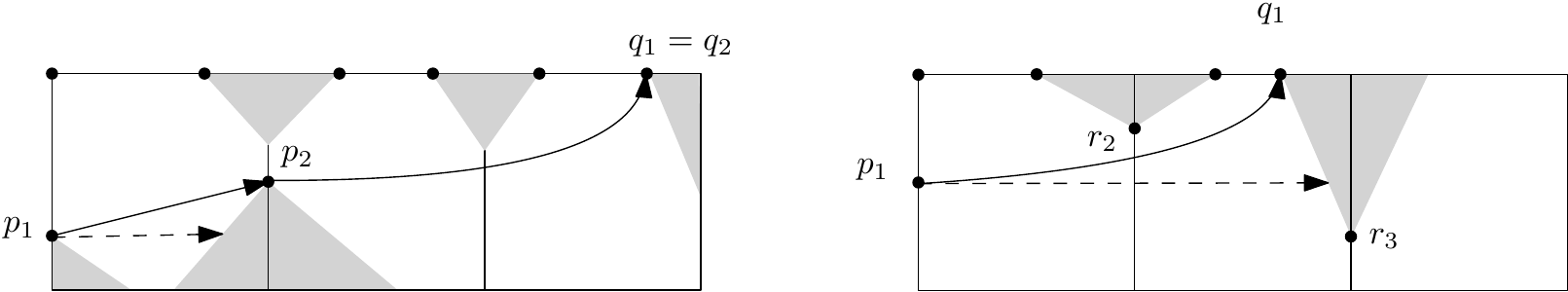}
    \caption{In the left figure, $p_1$ is blocked by $p_2$. In the right figure, $p_1$ is blocked by $r_3$.}
    \label{fig:pi_blocked}
\end{figure}

Knowing the above, given $p_i$, if we can find the first $p_j$ or $r_j$ that blocks $p_i$, we can find $q_i$. We can use a binary search tree to store the top and bottom critical points of the vertical cell boundaries based on their increasing $x$-coordinates, where critical points on the same vertical cell boundary are stored in the same leaf. In addition, each node $v$ stores the maximum and minimum $y$-coordinates of the critical points in $v$'s canonical subset. One can query such binary search tree to find the first $p_j$ or $r_j$ that blocks $p_i$ in $O(\log n)$ time. In fact, we can do the same thing for any arbitrary point $p$ in the free space as long as we know which cell $p$ reside in.

\begin{observation}
\label{observation:find_pj}
Given a set of $n_k$ consecutive non-empty cells in the same row with vertical boundary indices $\{1, ..., n_k + 1\}$, one can preprocess them in $O(n_k \log n_k)$ time such that one can find the first $p_j$ or $r_j$ that blocks $p$ in $O(\log n)$ time. 
\end{observation}

With the above observation, and the dynamic programming algorithm described by Gudmundsson and Wong \cite[Lemma~20]{GudmundssonWong2022Cubicupperlower}, we can preprocess the free space diagram such that given a free point~$p$, we can find the rightmost critical point on the top boundary of the same row in $O(\log n)$ time. This property is useful in Section~\ref{sec:arbitrary_ref_trajectory} when the reference trajectory is no longer vertex-to-vertex, and a monotone path can start from the interior of a non-empty cell. We combine the above insights and observation in the below Lemma~\ref{lemma:find_qi}.  

\begin{lemma}
\label{lemma:find_qi}
Given a row of $n_k$ non-empty cells in the free space diagram and a point $p$ in the free space, one can preprocess these cells in $O(n_k\log n_k)$ time, and find the right-most critical point $q$ on the top boundary such that there is a $pq$ monotone path in $O(\log n_k)$ time.
\end{lemma}

\begin{proof}
Given a row of non-empty cells, we use the below algorithm on each set of consecutive adjacent non-empty cells such that the boundary between two cells intersects free space. Let $x(p)$ and $y(p)$ be the $x$ and $y$-coordinate of point $p$, respectively. Let the index of the leftmost and rightmost vertical boundary be $l$ and $r$, respectively. 

Preprocess the row using Observation~\ref{observation:find_pj} in $O(n_k \log n_k)$ time. Iterate $i$ from $r$ to $l$, and for each $p_i$, use Observation~\ref{observation:find_pj} to find the leftmost $p_j$ or $r_j$ such that $y(p_i) \leq y(p_j)$ or $y(p_i) \geq y(r_j)$, respectively, in $O(\log n_k)$ time. If $p_j$ is leftmost, set $q_i = q_j$. If $r_j$ is leftmost, use binary search to find the rightmost critical point $q$ such that $x(p_i) \leq x(q) \leq x(r_j)$ in $O(\log n_k)$ time. In total, this takes $O(n_k \log n_k)$ time. Then, given a free point $p$ between the $i$th and the $(i + 1)$th cell boundaries, one can find the first $p_j$ or $r_j$ that blocks $p_i$, and repeat the above process to find the rightmost critical point $q$ such that there is a $pq$ monotone path in $O(\log n_k)$ time.

The correctness of this algorithm relies on two key facts. One, if $p_j$ is the leftmost critical point such that $p_i$ is blocked by $p_j$, then there is a $p_ip_j$ monotone path. Second, if $r_j$ is the leftmost critical point that blocks $p_i$, then there is a $pq$ monotone path as long as $q$ is between $p_i$ and $r_j$. Indeed, since the intersection of the free space and a cell is convex, if $p_i$ is not blocked by any critical point of a cell $C$, $p_i$ cannot be blocked by the interior of $C$. 
\end{proof}

\section{Proof of Lemma~\ref{lemma:cells_constant}}
\label{appendix:cells_constant}

\cellsconstant*

\begin{proof}
    The following definition follows from the definition of $f_{P, \ed}$. Let $u'$ be the first intersection of $B(u, \varepsilon d)$ and $(u, v)$ along $simpl(P, \ed)$, and let $a'$ be the first intersection of $B(a, \varepsilon d)$ and $(a, b)$ along $simpl(Q, \ed)$. Let $w$ be the first intersection of $B(u, \varepsilon d)$ and $P_{uv}$, and let $c$ be the first intersection of $B(a, \varepsilon d)$ and $Q_{ab}$. 

    Consider a partition of $(P_{uv}, Q_{ab})$ cells into four cells generated from $P_{uw}$, $P_{wv}$, $Q_{ac}$, and $Q_{cb}$. Define $W$ to be the rectangle $[0, \abse{P_{uw}}] \times [0, \abse{Q_{ac}}]$. We will show that the intersection of the simplified free space and $(P_{uw}, Q_{ac})$ cells is an ellipse clipped at $W$ by using the arguments in the proof of \cite[Lemma~30.2.1]{Har-peled2011GeometricApproximationAlgorithms}.
    
    With slight abuse of notation, we redefine $P_{uw}$ as the affine mapping from $[0, \abse{P_{uw}}]$ to points on $P_{uw}$. The function $f_{P, \ed} \circ P_{uw}$ is an affine function as the composition of two affine functions $f_{P, \ed}$ and $P_{uw}$ is also affine, and similarly, $f_{Q, \ed} \circ Q_{ab}$ is affine. Therefore $h: (x, y) \mapsto f_{P, \ed}(P_{uw}(x)) - f_{Q, \ed}(Q_{ab}(y))$ is an affine function. And assuming general position of $(u, v)$ and $(a, b)$, $h$ is also one-to-one. All the desired configurations of $(x, y)$ values satisfying $\abse{h(x, y)} \leq d$ are mapped to the disk $B$ with radius $d$ centered at the origin. 
    
    Consider the intersection of $B$ and the image of $h$, i.e., $h(\mathbb{R}^2) \cap B$, then the simplified free space of the $(P_{uw}, Q_{ac})$ cells is the set $h^{-1}(h(\mathbb{R}^2) \cap B) \cap W$. The inverse of an affine function is also an affine function, and as such so is $h^{-1}$. The affine image of a disk is an ellipse, and as such so is $h^{-1}(h(\mathbb{R}^2) \cap B)$. Therefore the $(P_{uw}, Q_{ac})$ cell is an ellipse clipped into a rectangle, and it takes constant time to compute.

    The above arguments can be extended to the other three cells. Therefore, constructing the cells defined by $P_{uw}$ and $Q_{ab}$ in $\mathcal{F}'_{(1 + \e)d}(P, Q)$ takes constant time. 
\end{proof}

% \section{Proof of Lemma~\ref{lemma:lambda_bound_by_c_e}}
% \label{appendix:lambda_bound_by_c_e}

% \lambdabound*

% \begin{proof}
% The set $C_u$ of four canonical squares cover $S_u$, and let $b$ be the intersection of these squares. Construct a ball $B$ such that $B$ is centered at $b$, and the minimum distance from the border of $C_u$ to the border of $B$ is $\abse{u}$. \citet{gudmundssonApproximatingLambdalowdensityValue2022} showed that the diameter of $C_u$ is at most a constant times the diameter of $S_u$, say $a \cdot ((1 + \ehat)d + \abse{u})$, where $a$ is a constant. Therefore, the radius of $B$ is at most $a \cdot ((1 + \ehat)d + \abse{u}) + \abse{u}$. The number of segments $v$ with $\abse{v} \geq \abse{u}$ intersecting $C_u$ can thus be bounded as follow. 

% \begin{align*}
%     \#segments &\leq \frac{\text{the maximum $\abse{P' \cap B}$}}{\text{the minimum length of a segment}} \\
%     &= \frac{6c \cdot (a \cdot ((1 + \ehat)d + \abse{u}) + \abse{u})}{\abse{u}} && \text{($P'$ is $6c$-packed.)} \\
%     &\leq \frac{6c \cdot a \cdot d}{\ed} + O(c) && \text{($\abse{u} \geq \ed$)} \\
%     &\in O(c/\varepsilon)
% \end{align*}
% \end{proof}

\section{Handling Additional Critical Points}
\label{appendix:greedy_crit_points}
There are additional critical points to consider. With monotone path $P_i$ computed, the start of $P_{i + 1}$ should be the lowest feasible free point $p$ on $l_s$. Similarly, the finishing point of a monotone path should be the lowest feasible free point $p'$ on $l_t$. We will call these points $p$ and $p'$ the \textit{greedy critical points}, and they may not be in $G$. Gudmundsson and Wong~\cite{GudmundssonWong2022Cubicupperlower} showed that the number of greedy critical points is bounded by $O(nm)$, that is, $m$ monotone paths for each of the $O(n)$ reference trajectories.

% To bound the number of greedy critical points in the free space diagram of two simplified $c$-packed trajectories, we show $m \in O(c/\varepsilon)$. The idea is similar to that of Lemma~\ref{lemma:lambda_bound_by_c_e}.
    
% \begin{lemma}
% \label{lemma:bound_m_with_c}
% Let $T$ be the $\ed$-simplification of a $c$-packed trajectory. If there is a solution $\{T_1, ..., T_m\}$ to $\mathtt{SC(T, m, l, (1 + \varepsilon)d)}$, then $m \in O(c /\varepsilon)$.
% \end{lemma}

% \begin{proof}
% Let $u$ be the shortest segment in the trajectories $\{T_1, ..., T_m\}$, and let $B$ be a ball centered at the midpoint of $u$ with radius $3\abse{u}/2 + d$, and let $S_{u}$ be the stadium of radius $d$ around $u$. Observe that the distance from the border of $P_{u}$ to the border of $B$ is $\abse{u}$. Therefore a segment $s_i \in T'$ that intersect $S_u$ must intersect $B$, and the length of their intersection must be at least $\abse{u}$. The total number of segments intersecting $S_{u}$ is
% \begin{align*}
%     \text{\#segments} &\leq \frac{\text{maximum length of } \abse{T' \cap B}}{\text{minimum length of segments covered by $B$ }} \\
%     &= \frac{6c \cdot (\frac{3}{2}\abse{u} + d)}{\abse{u}}\\
%     &= 9c + 6c \cdot \frac{d}{\abse{u}}\\
%     &\leq 9c + 6c \cdot \frac{d}{\ed} \\
%     &= 9c + \frac{6c}{\varepsilon}
% \end{align*}
% In the worst case, each segment belongs to a different subtrajectory, and $m \in O(c / \varepsilon)$.
% \end{proof}

We also need to add the greedy critical points in $G$. Recall that in Lemma~\ref{lemma:find_qi}, we describe an algorithm to compute, for a point $p$ in the free space, the rightmost critical point $q$ on the top boundary such that there is a $pq$ monotone path in $O(\log n_k)$ time if $p$ belongs to a group of $n_k$ consecutive non-empty cells. We have also shown that there are only $O(c/\varepsilon)$ critical points on a horizontal boundary. We can now bound the number of greedy critical points, the time it takes to insert them into $G$, and the total number of edges of $G$ with the additional greedy critical points. 

% \begin{lemma}
% \label{lemma:number_of_greedy_crit_points}
% In the simplified free space diagram $\mathcal{F}'_{(1 + \ehat)d}(T, T)$, we generate at most $O(cn/\varepsilon)$ greedy critical points. It takes $O((cn/\varepsilon)\log(n))$ time to add them in the graph $G$, and $G$ has at most $O((cn/\varepsilon)\log(c/\varepsilon))$ edges.
% \end{lemma}

\begin{lemma}
\label{lemma:number_of_greedy_crit_points}
In the simplified free space diagram $\mathcal{F}'_{(1 + \ehat)d}(T, T)$, we generate at most $O(nm)$ greedy critical points. It takes $O(nm\log n)$ time to add them in the graph $G$, and $G$ has at most $O(nm\log (c/\e))$ edges.
\end{lemma}

\begin{proof}
There are at most $O(nm)$ greedy critical points \cite{GudmundssonWong2022Cubicupperlower}. For each greedy critical point $p$, it takes $O(\log n)$ amortised time to find the rightmost critical point $q$ such that there is a $pq$ basic monotone path, and it takes $O(\log(c/\varepsilon))$ time to range query the $O(\log(c/\varepsilon))$ nodes to connect $p$ to in the range tree. In total, it takes $O(nm\log n)$ time to insert the greedy critical points in $G$. Each $p$ connects to at most $O(\log{(c/\varepsilon)})$ nodes, therefore we add at most $O(nm \log(c/\varepsilon))$ edges. The graph $G$ has $O(nm\log(c/\varepsilon))$ edges in total.
\end{proof}

\section{Proof of Lemma~\ref{lemma:arbitrary_reference_no_interval_management}}
\label{appendix:arbitrary_reference_no_interval_management}

\arbitraryRefNoIntervalManagement*

\begin{proof}
Construction of the simplified free space diagram takes $O(cn/\varepsilon + cn/\varepsilon^3)$ by Theorem~\ref{thm:free_space_summarised}. Construction of $G$ takes $O((cn/\varepsilon)\log{(n/\varepsilon)})$ time by Lemma~\ref{lemma:constructing_G}. The total number of critical points is upperbounded by the number of propagated critical points, which is $O(c^2 n/\varepsilon^2)$. The total number of greedy critical point is $O((c^2n/\varepsilon^2) \cdot m)$. Sorting the critical points takes $O((c^2mn/\e^2)\log(n/\varepsilon))$ time. For every critical point $p$, compute the rightmost point $q$ such that there is a $pq$ monotone path takes  $O((c^2nm/\varepsilon^2) \log n)$ time by Lemma~\ref{lemma:find_qi}. The graph $G$ has at most $O((c^2 mn/\varepsilon^2)\log(c/\varepsilon))$ edges after adding the greedy critical points, and the internal critical points. Gudmundsson and Wong~\cite{GudmundssonWong2022Cubicupperlower} showed that an edge is added to and removed from the link-cut tree at most once, and adding/removing an edge from the link-cut tree takes $O(\log(n/\varepsilon))$ time since the maximum number of nodes in the link-cut tree is upperbounded by the number of nodes in $G$. Therefore maintaining the link-cut tree takes $O((c^2 mn/\varepsilon^2)\log(c/\varepsilon)\log(n/\varepsilon))$ time.
\end{proof}

\section{Infinite $l$-apart Critical Points} \label{appendix:l_apart}

One can detect that two cells generate infinite $l$-apart critical points by comparing the quadratic equations of their interiors (see Figure~\ref{fig:infinite_l_apart}). To cope with this case, one can move the end point of one of the segment by a miniscule amount, say $\delta = \e^2 d$.

\begin{figure}[bth]
    \centering
    \includegraphics[scale=0.8]{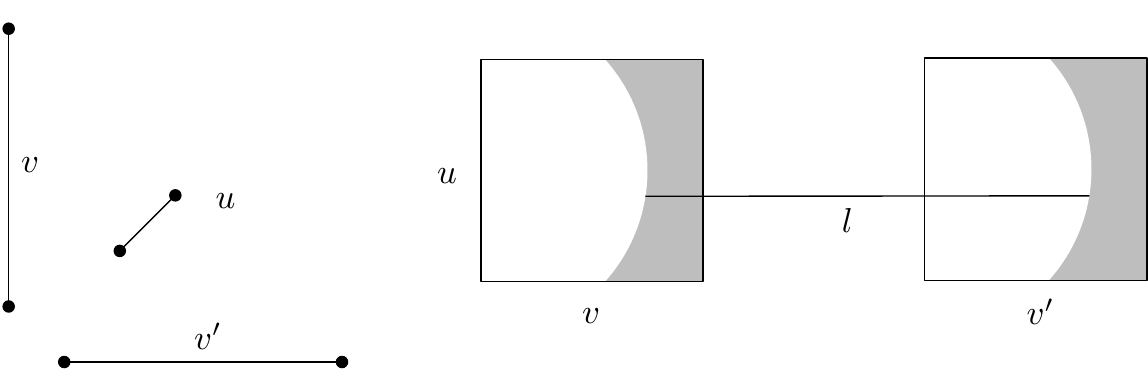}
    \caption{The placement of segments $u$, $v$, and $v'$ and their respective $(u, v)$ and $(u, v')$ cells. There are an infinite number of $l$-apart critical points when the boundaries of the free space in two cells can be expressed with the same equation, and they are exactly $l$-apart.}
    \label{fig:infinite_l_apart}
\end{figure}
\end{document}